\documentclass[a4paper, thm-restate]{lipics-v2021}
\usepackage[utf8]{inputenc}
\usepackage{amsmath}
\usepackage{amsfonts}
\usepackage{amsthm}
\usepackage{graphicx}
\usepackage{wrapfig}
\usepackage{makecell}
\usepackage{xspace}
\usepackage{caption}
\usepackage{imakeidx}
\usepackage{thmtools}
\usepackage{todonotes}
\usepackage{stmaryrd}
\usepackage{mathtools}
\usepackage{xcolor}
\usepackage{booktabs}
\usepackage{comment}

\hideLIPIcs
\nolinenumbers

\newcommand{\etal}{et al.\xspace}
\newcommand{\R}{\mathbb{R}}
\renewcommand{\S}{\mathcal{S}}
\newcommand{\D}{\ensuremath{\mathcal{D}}\xspace}
\newcommand{\RR}{\ensuremath{\mathcal{R}}\xspace}
\newcommand{\MBM}{\ensuremath{\mathrm{MBM}}\xspace}

\newcommand{\eps}{\varepsilon}

\graphicspath{{./Figures/}}

\title{Fully-Adaptive Dynamic Connectivity of Square Intersection Graphs}
\titlerunning{Fully-Adaptive Dynamic Connectivity of Square Intersection Graphs}
\author{Ivor van der Hoog}{Technical University of Denmark, Denmark}{idjva@dtu.dk}{https://orcid.org/0009-0006-2624-0231}{This project has additionally received funding from the European Union's Horizon 2020 research and innovation programme under the Marie Sk\l{}odowska-Curie grant agreement No 899987.}

\author{Andr\'{e} Nusser}{
CNRS, Inria, I3S, Université Côte d'Azur, France}{andre.nusser@cnrs.fr}{}{This work was supported by the French government through the France 2030 investment plan managed by the National Research Agency (ANR), as part of the Initiative of Excellence of Université Côte d’Azur under reference number ANR-15-IDEX-01.}
\author{Eva Rotenberg}{Technical University of Denmark, Denmark}{erot@dtu.dk}{https://orcid.org/0000-0001-5853-7909}{Supported by the Independent Research Fund Denmark grant
2020-2023 (9131-00044B) ``Dynamic Network Analysis'' and the
Carlsberg Foundation Young Researcher Fellowship CF21-0302 ``Graph
Algorithms with Geometric Applications''.}
\author{Frank Staals}{Utrecht University, Netherlands}{f.staals@uu.nl}{}{}

\authorrunning{Ivor van der Hoog, Andr\'{e} Nusser, Eva Rotenberg, Frank Staals}
\Copyright{Ivor van der Hoog, Andr\'{e} Nusser, Eva Rotenberg, Frank Staals}

\relatedversion{}

\ccsdesc[500]{ Theory of computation ~ Design and analysis of algorithms}
\ccsdesc[100]{Theory of computation~Dynamic graph algorithms}
\ccsdesc[100]{Theory of computation~Computational Geometry}

\keywords{Computational geometry, planar geometry, data structures, geometric intersection graphs, fully-dynamic algorithms}

\begin{document}

\maketitle

\begin{abstract}
A classical problem in computational geometry and graph algorithms is: given a dynamic set $\mathcal{S}$ of geometric shapes in the plane, efficiently maintain the connectivity of the intersection graph of $\mathcal{S}$.
Previous papers studied the setting where, before the updates, the data structure receives some parameter $P$. 
Then, updates could insert and delete disks as long as at all times the disks have a diameter that lies in a fixed range $[\frac{1}{P}, 1]$. 
As a consequence of that prerequisite, the aspect ratio $\psi$ (i.e. the ratio between the largest and smallest diameter) of the disks would at all times 
satisfy $\psi \leq P$. 
The state-of-the-art for storing disks in a dynamic connectivity data structure is a data structure that uses $O(Pn)$ space and that has amortized $O(P \log^4 n)$ expected amortized update time. Connectivity queries between disks are supported in $O(\log n / \log \log n)$ time. 
The state-of-the-art for Euclidean disks immediately implies 
a data structure for connectivity between axis-aligned squares that have their diameter in the fixed range  $[\frac{1}{P}, 1]$, with an improved update time of $O(P \log^4 n)$ amortized time. 

In the dynamic setting, one wishes for a more flexible data structure in which disks of any diameter may arrive and leave, independent of their diameter, changing the aspect ratio freely. 
We study fully-dynamic square intersection graph connectivity. Our result is fully-adaptive to the aspect ratio, spending time proportional to the current aspect ratio $\psi$, as opposed to some previously given maximum $P$. 
Our focus on squares allows us to simplify and streamline the connectivity pipeline from previous work.
When $n$ is the number of squares and $\psi$ is the aspect ratio after insertion (or before deletion), our data structure answers connectivity queries in $O(\log n / \log \log n)$ time.
We can update connectivity information in $O(\psi \log^4 n + \log^6 n)$ amortized time.
We also improve space usage from $O(P \cdot n \log n)$ to $O(n \log^3 n \log \psi)$ --- while generalizing to a fully-adaptive aspect ratio --- which yields a space usage that is  near-linear in $n$ for any polynomially bounded $\psi$.
\end{abstract}

\newpage
\section{Introduction}
\setcounter{page}{1}

Geometric intersection graphs are one of the most well-studied geometrically-flavoured graph classes:
Their nodes are geometric shapes, and an edge between two such shapes exists if and only if they intersect.
This makes the description complexity of a geometric intersection graph very compact; it is linear in the number of objects, while the underlying graph potentially has a quadratic number of edges.
In this work, we consider square intersection graphs in the dynamic setting.
Intersection graphs are one of the few examples of dynamic graphs where fully-dynamic insertion and deletion of vertices is motivated and interesting. Since a vertex can come or leave with $\Theta(n)$ edges, applying any existing edge-updatable dynamic graph algorithm in a blackbox manner would lead to $\Omega(n)$ update time. Yet, the geometric nature of these graphs often allows for sublinear or even polylogarithmic update times. 

A classical problem in dynamic graph algorithms is dynamic connectivity.
In this problem, we want to maintain a data structure under edge or node insertions and deletions that allows for fast queries that return whether a given pair of vertices is connected. Connectivity was one of the first problems to be studied in the dynamic setting dating back to the 80s~\cite{Frederickson85, SleatorT83}, and has received ample attention ever since.
Dynamic connectivity in general graphs has been studied in many settings; randomised or deterministic, amortised or worst-case~\cite{Frederickson85,HenzingerK99,Thorup00connectivity,KapronKM13,NanongkaiSW17}, and the partially dynamic incremental or decremental settings~\cite{Tarjan79,TarjanL84,Thorup99,Aamand}. Due to its fundamental nature, and its many applications, dynamic connectivity has also received much attention for simpler graph classes. Examples of such graph classes include trees~\cite{SleatorT83,AlstrupSS97}, planar graphs~\cite{eppstein1993separator,LackiS15}, and graphs of bounded genus~\cite{Eppstein03,HolmR21}.

Naturally, the dynamic connectivity problem also drew attention for the class of geometric intersection graphs.
This setting is particularly interesting as a single node insertion can drastically change the number of connected components, as it can introduce a linear number of new edges.
On the other hand, the geometric structure can be exploited in the data structures.
The first result for geometric connectivity in geometric intersection graphs with update and query time independent of the object diameters is by Chan \etal~\cite{CPR11} who presents a dynamic (Euclidean) disk intersection data structure with an update time of $O(n^{20/21 + \eps})$ and a query time of $O(n^{1/7 + \eps})$. This has recently been improved to $O(n^{7/8})$ amortized update time with constant query time~\cite{chan2024dynamic}.
As progress seemed difficult in this setting, the setting in which the disks in the data structure have restricted diameters was considered.
For a fixed diameter range, where there is some value $P$ given in advance and diameters have to be contained in an interval $[\frac{1}{P}, 1]$, Kaplan \etal~\cite{kaplan2020dynamic} showed that there is a data structure with expected amortized $O(P^2 \log^{10} n)$ update time, query time $O(\log n/ \log \log n)$, using $O(n P \log n)$ space.
Recently, Kaplan \etal~\cite{kaplan2022dynamic} improved this to $O(P \log^7 n)$ expected amortized update time with the same update time and $O(n P)$ space. 

\subparagraph{From disks to squares.}
While the above works are stated for Euclidean disks, we note that the approach in \cite{kaplan2022dynamic} can be combined with \cite{willard1985adding} to obtain a data structure that works for the simpler setting of connectivity between axis-aligned squares. The obtained update time is amortized $O(P \log^4 n)$.
Disk intersection graphs are often motivated by communication networks where the disks are interpreted as some sort of transmission diameter. This is an idealization of a complicated physical process and actual ad-hoc communication networks do not correspond to perfectly circular disks~\cite{adhocnets}. Thus, it is reasonable to switch to a different metric for computational reasons while maintaining the core idea of the underlying problem.
Recently, similar progress was made for computing a single-source shortest path tree in an intersection graph by assuming square regions instead of disks~\cite{klost2023algorithmic}. 

\newpage
We study the dynamic connectivity problem where the input $S$ is a set of (axis-aligned) squares,  while being fully adaptive to their aspect ratio. 
We let $\psi$ denote the adaptive aspect ratio.
Formally, if $S$ is the input before an update and $S'$ is the input after an update we define $\psi =  \max \{ \frac{\max_{\sigma \in S} |\sigma|}{\min_{\sigma' \in S} |\sigma'| }, \frac{\max_{\sigma \in S'} |\sigma|}{\min_{\sigma' \in S'} |\sigma'|} \}$. 
Our date structure maintains connectivity between squares with $O(\psi \log^4 n + \log^6 n)$ update time, $O(\log n / \log \log n)$ query time, and using $O(n \log^3 n \log \psi)$ space, see Table~\ref{tab:results} for a comparison.
our approach only requires near-linear space while maintaining near-linear update time and polylogarithmic query time.

\subparagraph*{Implications of adaptivity and our reduced space usage. }
To understand the implications of 
the adaptivity of our new solution,
consider 
the scenario where the set of squares starts with a square $A$ with diameter $1$ and 
$B$ with diameter $\frac{1}{n}$. Now, suppose the sequence of updates first removes $B$, then inserts a sequence of $n$ squares of diameter $1$, and finally reinserts $B$. 
Previous work~\cite{kaplan2022dynamic} would require as input the interval $[\frac{1}{n}, 1]$ and the promise that all updates only insert squares in this interval. 
The space usage and the total update time of \cite{kaplan2022dynamic} is quadratic. 
In our case, since for almost all updates, the aspect ratio is constant. Moreover, the space usage and total update time is near-linear.

\begin{table}
    \caption{Complexities are asymptotic. All update times are amortized. The query time for all approaches is $O(\log n / \log \log n)$. $\lambda_s(n)$ denotes the maximum length of a Davenport-Schinzel
sequence of order $s$ on $n$ symbols. }
    \label{tab:results}
    \centering
    \begin{tabular}{ c lllll}
    \toprule
      Object  & Aspect ratio & Update time & Space & Ref. \\
        \midrule
         disks  & fixed $[\frac{1}{P}, 1]$ & $P \cdot \log^7 n \cdot \lambda_6(\log n)$\, exp. & $n P$ & \cite{kaplan2022dynamic} \\
disks & fixed $[\frac{1}{P}, 1]$ & $P \cdot \log^4$\, exp. & $n P$ & \cite{kaplan2022dynamic}+\cite{liu2022nearly} \\        
squares  & fixed $[\frac{1}{P}, 1]$ & $P \cdot \log^4 n$ & $n P$ & \cite{kaplan2022dynamic}+\cite{willard1985adding} \\
         squares & adaptive $\psi$ & $ \psi \log^4 n + \log^6 n $ & $n \log^3 n \log \psi$ & Thm~\ref{thm:L1}  \\
         \bottomrule
    \end{tabular}
\end{table}

\section{Problem statement and technical overview}
Let $\S \subset \R^2$ be a set of axis-aligned squares.
The intersection graph $G[\S]$ is the graph with vertex set $\S$ and with an edge between squares $\sigma, \sigma' \in \S$ whenever they intersect. 
We say that two squares $\sigma, \sigma'$ are \emph{connected} if there exists a path between their corresponding vertices in $G[\S]$.
The set $\S$ is a fully dynamic set subject to (adversarial) insertions and deletions of squares.
We wish to maintain $\S$ in a data structure supporting \emph{connectivity queries} between squares in $\S$. We denote the diameter of square $\sigma \in \S$ by $|\sigma|$.
We consider three settings with different restrictions on the square diameters in $\S$:
\begin{itemize}
\item The \emph{fixed diameter range} setting. Here, the input specifies some $P$ and  each $\sigma \in \S$ has a diameter in $[\frac{1}{P}, 1]$ for some fixed $P$. 
\item The \emph{bounded aspect ratio} setting. Here, the input specifies some $P$ and at all times, for all $\sigma, \sigma' \in \S$ we have $\frac{|\sigma|}{|\sigma'|} \leq P$. 
\item The \emph{adaptive aspect ratio} $\psi$ setting in which arbitrary insertions and deletions may occur. Let $S$ be the set of squares before an update and $S'$ be the set of squares after an update. We define $\psi =  \max \{ \frac{\max_{\sigma \in S} |\sigma|}{\min_{\sigma' \in S} |\sigma'|}, \frac{\max_{\sigma \in S'} |\sigma|}{\min_{\sigma' \in S'} |\sigma'|} \}$ the 
aspect ratio relevant for the update.  
\end{itemize}
\noindent
We measure the algorithmic complexity in $n \coloneqq |\S|$ and $\psi$, where $n$ is the present size of the dynamic set $\S$. 
At all times, we maintain some minimal axis-aligned square $F$ that contains $\S$, and the coordinates of $F$ are powers of $2$. 

\subparagraph{Results.} \label{sec:results}
 The data structure of Kaplan~\etal~\cite{kaplan2022dynamic}, when adapted to axis-aligned squares by applying Range Trees~\cite{willard1985adding}, can store a dynamic set of axis-aligned squares as follows: The input specifies some value $P$ and all squares have fixed diameter range with ratio $P$. Their solution uses $O(n P)$ space and supports updates to $\S$ in $O(P \log^4 n)$ amortized  time (Table~\ref{tab:results}). 
They answer connectivity queries in $O(\log n / \log \log n)$ worst-case time.

There are several reasons why \cite{kaplan2020dynamic} uses $O(n P)$ space and  why their update bound cannot depend on the adaptive $\psi$ instead of $P$ (which we discuss further down). 
We present an adaption of their work that relies on the fact that $\S$ is a set of axis-aligned squares. Under this assumption, we adapt their data structure to work for adaptive $\psi$. We improve the space usage to near linear in $n$ and $\log \psi$. 
In full generality, we also improve the performance: allowing for update times proportional to the density around the update $\sigma$.
 More concretely, we parameterize the runtime by the size of two sets $\mathcal{C}(\sigma)$ and $\mathcal{P}(\sigma)$. Intuitively, these sets contain squares in $\sigma$ or squares around $\sigma$. These sets have size at most $O(\min\{\psi,n\})$.

 Our result is a technical contribution, that examines and refines the data structure in~\cite{kaplan2022dynamic} in the special case where $\S$ is a set of axis-aligned squares. 
To detail our contribution, we now present a technical overview, where we reference several concepts whose formal definitions are presented in their respective sections.
The core component is a quadtree that stores $\S$.

\subparagraph*{The existing pipeline.}
We can describe the data structure of~\cite{kaplan2022dynamic} (adapted to squares) on a high-level:  
They construct a quadtree $H(\S)$ in which quadtree cells may store squares in $\S$.
For any quadtree cell $C$, we denote by $\pi(C)$ the set of squares stored in $C$.
A crucial definition is the concept of a \emph{perimeter}. 
For a square $\sigma$, its perimeter $\mathbb{P}^*(\sigma, P)$ is intuitively a ring of $\Theta(P)$ quadtree cells of size at least $\frac{1}{4 P}$ that are sufficiently close to the boundary of $\sigma$. 
Using this concept, their data structure is a pipeline of five components (Fig~\ref{fig:pipeline}):
\begin{enumerate}
\item The set $\S$ gets stored in a compressed quadtree $H(\S)$, such that for every $\sigma \in \S$, the quadtree contains $\mathbb{P}^*(\sigma, P)$.\footnote{Whilst originally their approach is a forest of quadtrees, we note that since each root of the forest is disjoint, the whole solution can be stored as a quadtree. }
This quadtree has $\Theta(P \cdot n)$ cells. For each $\sigma \in \S$, its storing cell is the maximal quadtree cell contained in $\sigma$.
\item They store all (maximal) quadtree cells contained in some $\sigma \in \S$ in a special ancestor data structure. We leave out the details of this structure, as we show that it suffices to use the well-studied marked-ancestor data structure by Alstrup, and Husfeldt, Rauhe~\cite{alstrup1998marked}.  
\item For each square $\sigma \in \S$ with storing cell $C_\sigma$, for each quadtree cell $C_2 \in \mathbb{P}^*(\sigma, P)$, 
they store a square intersection data structure (a range tree).
This data structure stores the squares $R \subset \pi(C_\sigma)$ that have $C_2$ in their perimeter (i.e. $R = \{ \gamma \in \pi(C_\sigma) \mid C_2 \in \mathbb{P}^*(\sigma, P) \}$).  
\item  For each square $\sigma$ with storing cell $C_\sigma$, for each quadtree cell $C_2 \in \mathbb{P}^*(\sigma, P)$, they store a maximal bichromatic matching (\MBM) in the graph $G[R \cup B]$ with $B = \pi(C_2)$ with $R = \{ \gamma \in \pi(C_\sigma) \mid C_2 \in \mathbb{P}^*(\sigma, P) \}$.
\item They store a \emph{proxy graph} over the quadtree in the dynamic connectivity data structure by Holm, Lichtenberg, and Thorup ($HLT$) \cite{holm2001poly}. This graph contains an edge between two cells $C_1, C_2$ if and only if their maximal bichromatic matching is not empty. 
\end{enumerate}%
\noindent
They subsequently support connectivity queries for a query $(\sigma, \rho)$ as follows. 
    Given $\sigma$, obtain a pointer to its storing cell $C_\sigma$. 
    Using their ancestor data structure, they obtain the largest ancestor $C_a$ that is contained in a square $\sigma^* \in \S$. Let $C^*$ be its storing cell. 
    Doing the same procedure for $\rho$ gives a cell $R^*$. 
    They show that $(\sigma, \rho)$ are connected in $G[\S]$ if and only if $(C^*, R^*)$ are connected in the proxy graph; which they test in  $O(\log n / \log \log n)$ time. 
\begin{figure}[b]
  \centering
  \includegraphics[page=4]{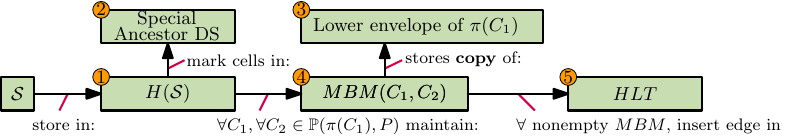}
  \caption{
  The five-component pipeline by Kaplan~\etal where the arrows indicate dependencies.
  }
  \label{fig:pipeline}
\end{figure}
\subparagraph{Why the space usage is high, and update times are not adaptive.}
For every $\sigma \in \S$, this pipeline creates a set of $\Theta(P)$ quadtree cells with sizes in $[\frac{1}{4P}, 1]$ which we denote by $\mathbb{P}^*(\sigma, P)$.
For each quadtree cell $C_2 \in \mathbb{P}^*(\sigma, P)$, $\sigma$ gets stored in a square intersection data structure, even if the quadtree cell $C_2$ is empty and/or has no other squares nearby.
This approach makes it require a lot of space. Moreover, this approach fails when the aspect ratio becomes adaptive: suppose that $\S$ has $n-1$ squares with diameter $\frac{1}{P}$ and one unit square. We replace the unit square with a square of diameter $\frac{1}{P^2}$. 
The aspect ratio remains bounded by $P$. 
However (after rescaling the plane by a factor $\frac{1}{P}$) the quadtree no longer contains for every $\sigma \in \S$ the set $\mathbb{P}^*(\sigma, P)$ as  (after rescaling the plane) it only contains quadtree cells of size $1$. Reconstructing these requires $\Omega(P n)$ time. 

\subparagraph*{Our adaption.}
We improve this pipeline in several ways based on a few key insights:
First, we revisit the definition of \emph{perimeter}; presenting a new definition $\mathcal{P}(\sigma)$ which intuitively contains only cells in $\mathbb{P}^*(\sigma, \psi)$ that store at least one square. Since we only store data structures on cells that store at least one square, we save space and allow $\psi$ to become fully adaptive. This introduces a new challenge, as we need to work with significantly fewer precomputed inormation.  Concretely, we do the following (Figure~\ref{fig:newpipeline}):
\begin{enumerate}
    \item We define a new type of quadtree $T(\S)$ that uses only $O(n \log \psi)$ quadtree cells.
\item We replace their custom ancestor data structure by the
  well-studied \emph{M}arked \emph{A}ncestor \emph{T}ree (MAT), simplifying the
  data structure.
\item For squares metric we create a new data structure that has deterministic guarantees and that avoids storing many copies. 
\item  For each square $\sigma \in \S$ with storing cell $C_\sigma$, for
  each quadtree cell $C_2$ in our new perimeter $\mathcal{P}(\sigma)$,
  we store a Maximal Bichromatic Matching ($\MBM^*$) in a graph $G[R
  \cup B]$. Using our new definition of perimeter, we define $R \gets \{ \gamma \in \pi(C_\sigma) \mid C_2 \in \mathcal{P}(\sigma) \}$ and $B \gets \pi(C_2)$.
  We present a new algorithm to maintain this Maximal Bichromatic Matching. 
  \item Finally, we use HLT~\cite{holm2001poly} on a proxy graph with an edge for every non empty $\MBM^*$.
\end{enumerate}

\noindent
We go through our data structures one by one in the order indicated in Figure~\ref{fig:newpipeline}.

\begin{figure}
    \centering
  \includegraphics[page=5]{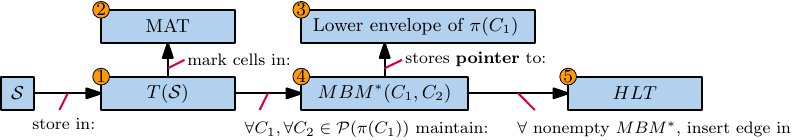}
  \caption{
  Our five-component pipeline where the arrows indicate dependencies.
  \vspace{-0.5cm}
  }
  \label{fig:newpipeline}
\end{figure}

\section{Storing disks in quadtrees}

Recall that $F$ is a (dynamic) square bounding box of $\S$. By construction, the side length of $F$ is $2^\omega$ for some integer $\omega$.
We define the square $F$ to be a \emph{quadtree cell} and we recursively define quadtree cells to be any square obtained by \emph{splitting} a cell into four equally sized closed squares.
A \emph{quadtree} $T$ on $F$ is any hierarchical decomposition obtained by recursively splitting cells. This hierarchical decomposition has a natural representation as a tree: the root is the cell $F$ and every cell has either $0$ or $4$ children depending on whether it was split. 
We denote by $\mathbb{F}$ the (infinite) set of cells that are obtained by recursively splitting all cells, starting from $F$. 
We say that a cell $C \in \mathbb{F}$ is at \emph{level} $\ell$ whenever its side length is $2^\ell$. 
We treat any quadtree $T$ as a set of cells, i.e., $T \subset \mathbb{F}$.

L\"{o}ffler, Simons, and Strash~\cite{loffler2013dynamic} use quadtrees to store arbitrary squares (or disks).  
Let $F$ be fixed and $\sigma$ be some square with center $s$. The unique \emph{storing cell} $C_\sigma \in \mathbb{F}$ is a largest cell in $\mathbb{F}$ that contains the center $s$ and that itself is contained in $\sigma$ (if $s$ lies at the intersection of multiple largest contained cells, we define the bottom left cell to be $C_\sigma$).
L\"{o}ffler, Simons and Strash subsequently say that $C_\sigma$ \emph{stores} $\sigma$. 
For a cell $C \in \mathbb{F}$, we denote by $\pi(C)$ all square of $\S$ stored in $C$. We will define five different cell sets to define our quadtree storing $\S$.

\subparagraph{Quadtree cell sets.}
We assume that we have some bounding box $F$ (which induces the set $\mathbb{F}$) and some set of storing cells in  $\mathbb{F}$.
We subsequently define two types of subsets of $\mathbb{F}$:
\begin{itemize}
    \item definitions that originate from~\cite{kaplan2022dynamic} and depend on some $P \in \mathbb{R}$  (blackboard font), 
    \item and new definitions depending on only the storing cells (calligraphic font).
\end{itemize}
Quadtree cells $C, C'$ are \emph{neighboring} whenever they are not descendants of one another and intersect in their boundary.
For any property, a quadtree cell $C$ in $\mathbb{F}$ is \emph{maximal} if there does not exist an ancestor of $C$ in $\mathbb{F}$ with the same property. 

\noindent
Let $\sigma \in \S$ have a storing cell $C_\sigma$. We define (Figure~\ref{fig:quadtree-def}):
\begin{itemize}
    \item $\mathcal{N}(\sigma) \subset \mathbb{F}$ as the cells of size $|C_\sigma|$ neighboring $C_\sigma$ or a neighbor of $C_\sigma$.
        \item $\mathbb{C}^*(\sigma, P) \subset \mathbb{F}$ as the maximal cells $C' \in \mathbb{F}$ with $C' \subset \sigma$  and $|C| \in [\frac{1}{4 P}, 1]$. 
    \item $\mathcal{C}(\sigma) \subset \mathbb{F}$ as the maximal cells $C' \in \mathbb{F}$ with $C' \subset \sigma$ that contain at least one storing cell.
    \item $\mathbb{P}^*(\sigma, P) \subset \mathbb{F}$ as the \emph{perimeter} of $\sigma$. These are all $C' \in \mathbb{F}$ contained in a cell in $\mathcal{N}(\sigma)$ with $|C'| \in [\frac{1}{4P}, 1]$ with the additional property that there exists a square $\rho \subset \mathbb{R}^2$ where $C'$ would be the storing cell of $\rho$ if $\rho \in S$, and, $\rho$ intersects the boundary of $\sigma$. 
    \item $\mathcal{P}(\sigma) \subset \mathbb{F}$ as all \emph{storing cells} with diameter at most $|\sigma|$ that, when scaled around their center by a factor 5, intersect \emph{the boundary} $\sigma$. Note that these may be contained in $\sigma$.
\end{itemize}
For $R \subseteq \S$, we define $\mathcal{N}(R)$ to be the union of all $\mathcal{N}(\sigma)$ with $\sigma \in R$. All other sets (e.g., $\mathcal{P}(R)$) are defined analogously.
Let $\ell_{\min}$ (resp. $\ell_{\max}$) be the smallest (resp. largest) level that contains any cell in any of the five sets. Per definition, $\ell_{\max} - \ell_{\min} \in O( \log \psi)$.

\begin{lemma}[Lemma~4.2 in \cite{kaplan2022dynamic}]
\label{lem:neighborhood_size}
For any $\sigma \in \S$, if regions are disks under an $L_p$ metric with a diameter in $[ \frac{1}{4 P}, 1]$ then:
$
|\mathbb{C}^*(\sigma, P)| \in O(P) \textnormal{ and } |\mathbb{P}^*(\sigma, P)| \in O(P).
$
\end{lemma}

\subparagraph{Compressed quadtrees.}
Denote by $X \subset \mathbb{F}$ some set of cells.
Denote by $T_X$ the minimal quadtree over some bounding box $F$ that contains all cells in $X$.
The size of $T_X$ can be arbitrarily large, even when $|X|$ is constant. 
To reduce quadtree space complexity, a quadtree may be \emph{compressed}~\cite{har2010quadtrees}. 
An $\alpha$-compressed quadtree (for some variable $\alpha \geq 1$) is defined as follows:
let $C$ be a quadtree cell in $T_X$ and $C_\alpha$ be the smallest descendant of $C$ such that (1) $|C| \geq 2^\alpha |C_\alpha|$ and (2) all cells in $X$ that are contained in $C$ are also contained in $C_\alpha$. Then $C$ has not $4$ children, but only $C_\alpha$ as its child. 
Given some constant $\alpha$, every quadtree has a unique maximally compressed equivalent that has size linear in $|X|$~\cite{har2010quadtrees}.
Given the above definitions, we want to mention two different quadtrees that store $\S$: 
\begin{itemize}
    \item \cite{loffler2013dynamic} defines $L(\S)$ as the compressed quadtree storing $\mathcal{N}(\S)$.
    \item \cite{kaplan2022dynamic} defines $H(\S)$ as the compressed quadtree storing $\mathcal{N}(\S)$, $\mathbb{C}^*(\S)$ and $\mathbb{P}^*(\S)$.   
\end{itemize}

\subparagraph*{Our quadtree.}
We define our quadtree $T(\S)$ as the compressed quadtree storing $\mathcal{N}(\S)$, where cells in $\mathcal{C}(\S)$ are uncompressed.
Note that any quadtree that contains $\mathcal{N}(\S)$, also contains the cells in $\mathcal{C}(\S)$.
The key difference between $L(\S)$ and $T(\S)$ is that we decompress the cells in $\mathcal{C}(\S)$, adding them to memory (i.e., we treat these cells as storing cells in the quadtree).
Since these cells are uncompressed, this structure uses more space than the $O(n)$ cells in $\mathcal{N}(\S)$. 
If we view quadtrees as a collection of (uncompressed) cells, then $L(\S) \subset T(\S) \subset H(\S)$. 
By Lemma~\ref{lem:neighborhood_size}, Kaplan \etal~\cite{kaplan2022dynamic} prove that $|\mathbb{C}^*(\S, P)|, |\mathbb{P}^*(\S, P)| \in O(P n)$. 
It would be easy to show that $|\mathcal{C}(\S)|, |\mathcal{P}(\S)| \in O(n \psi)$. But through clever counting, we prove that storing $T(\S)$ uses only $O(n \log \psi)$ space instead (Theorem~\ref{thm:quadtree_size}).

\begin{figure}
  \centering
  \includegraphics[width=\linewidth]{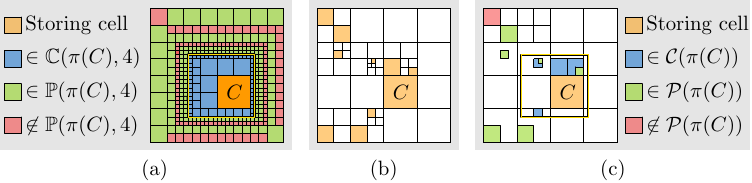}
  \caption{
      (a) We show for a square $\sigma$ its storing cell in orange. 
We set $P = \psi =  2$ and show our sets. 
Many cells in $\mathbb{C}^*(\pi(C), 8)$ are also in $\mathbb{P}^*(\pi(C), 8)$.
        (b) The minimal quadtree that contains a set of storing cells. 
        (c) Given the quadtree with storing cells, we illustrate our sets. Red cells are storing cells that occur in neither $\mathcal{C}(\pi(C))$ nor $\mathcal{P}(\pi(C))$. 
\vspace{-0.5cm}}
  \label{fig:quadtree-def}
\end{figure}

\subsection{Space complexity of the quadtree}
We upper bound the size of $\mathcal{N}(\S)$ and $\mathcal{C}(\S)$ (and thus the size of $T(\S)$).

\begin{restatable}{observation}{neighborhood}
    \label{obs:neighborhood}
    There are $O(n)$ cells in $\mathcal{N}(\S)$. 
\end{restatable}

\begin{restatable}{lemma}{uphillcontainment}
\label{lemma:uphill_containment}
    Let $C$ be a storing cell. 
    Denote by $Z$ any cell, such that there could exist a $\sigma$ stored in $Z$ where an ancestor of $C$ lies in $\mathcal{C}(\sigma)$. There are at most $O(\log \psi)$ such cells and we can report them in $O(\log \psi \log n)$ time.
\end{restatable}

\begin{proof}
    Let $|C| = 2^\ell$ (i.e., $C$ is at level $\ell$). 
By definition, $Z$ is in a level $j \geq \ell$.
Fix a level $j$.
For any cell $Z$ at level $j$,  $\mathcal{C}(\pi(Z) )$ contains an ancestor of $C$ only if a square in $\pi(Z)$ intersects (or contains) $C_j$ (the ancestor of $C$ at level $j$). 
As the diameter of squares in $\pi(Z)$ is at most a factor 5 larger than the diameter of $Z$, there are at most $O(1)$ cells at level $j$ that \emph{could} store a square $\rho$ that intersects $C_j$ (the neighbors of $C_j$, their neighbors and possibly their neighbors).
We can find these cells in $O(\log n)$ time by doing a point location in each cell for the level $j$.
The fact that per definition of $\psi$, all storing cells and all cells in $\mathcal{C}(\S)$ lie in a range of $O(\log \psi)$ levels concludes the proof.
\end{proof}

\begin{restatable}{theorem}{size}
\label{thm:quadtree_size}
    At all times, the compressed quadtree $T(\S)$ uses $O(n \log \psi)$ space. 
\end{restatable}

\begin{proof}
    Since $T(\S)$ is the quadtree that stores $\mathcal{N}(S)$ and $\mathcal{C}(S)$, and compressed quadtrees have linear space in the number of uncompressed cells, Observation~\ref{obs:neighborhood} and Lemmas~\ref{lemma:uphill_containment} immediately imply the theorem. 
\end{proof}

\noindent
We additionally upper bound the size of two more quadtree cell types:
\begin{restatable}{lemma}{perimetersize}
\label{lemma:perimeter_size}
Let $\S$ be a set of squares with aspect ratio $\psi$. For all $\sigma \in \S$ with storing cell $C_\sigma$ there are at most $O(\psi)$ cells in $\mathcal{C}(\sigma)$ and $\mathcal{P}(\sigma)$ and $O(\psi^2)$ cells in $\mathcal{P}(\pi(C_\sigma))$.
\end{restatable}

\begin{proof}
Consider any set of squares $S$. 
We rescale the plane such that the diameter of squares in $\S$ lies in $[\frac{1}{\psi}, 1]$.
We apply Lemma~\ref{lem:neighborhood_size} to conclude that $|\mathbb{C}(\sigma, \psi)|, |\mathbb{P}(\sigma, \psi)|\in O(\psi)$.
Note that the smallest storing cell in $T(\S)$ then has size  $\frac{1}{4 \psi}$.
Having rescaled, $\mathcal{C}(\sigma) \subseteq \mathbb{C}^*(\sigma, \phi)$.

Suppose that after rescaling there is a cell $C \in \mathcal{P}(\sigma)$ that is not in $\mathbb{P}^*(\sigma, \psi)$. 
Then $C$, scaled by a factor 5, intersects the boundary of $\sigma$.
Yet if $C \not \in \mathbb{P}^*(\sigma, \psi)$ then $C$ cannot store any square $\rho$ that intersects $\sigma$. 
Denote by $C'$ a neighbor of $C$ of size $2|C|$ that lies closer to the center of $\sigma$. 
It must be that $C' \in \mathbb{P}^*(\sigma, \psi)$ (indeed, we can construct a square with diameter $4|C|$ stored in $C'$ that intersects $\sigma$). 
This way, each cell in $\mathbb{P}^*(\sigma, \psi)$ can get charged by at most $O(1)$ cells $C' \in \mathcal{P}(\sigma)$ where $C' \not \in \mathbb{P}^*(\sigma, \psi)$. 
Thus, $|\mathcal{P}(\sigma)| \in O(\psi)$.
The upper bound on $|\mathcal{P}(\pi(C_\sigma))|$ follows  from the standard packing argument.
\end{proof}

\begin{restatable}{lemma}{uphillpermeter}
\label{lemma:uphill_permeter}
Let $C$ be a storing cell. 
Denote by $Z$ any cell, such that there could exist a $\sigma$ stored in $Z$ with $C \in \mathcal{P}(\sigma)$. We can report all $O(\log \psi)$ such cells in $O(\log \psi \log n)$ time.
\end{restatable}

\begin{proof}
Let $|C| = 2^\ell$ (i.e., $C$ is at level $\ell$ in the quadtree). 
By definition, every quadtree cell $Z$ of the lemma statement is stored at a level $j \geq \ell$.
Fix a level $j \geq \ell$ and let $C_j$ be the ancestor of $C$ at level $j$. 
If for any cell $Z$ at level $j$, $C \in \mathcal{P}(\pi(Z) )$ then it must be that the cells $Z$ and $C_j$ (when both are scaled around their center by a factor $5$) intersect. 
There are at most $O(1)$ such cells $Z$ at level $j$ for which this can be true. 
We can find these cells at level $j$ in $O(\log n)$ time by performing $O(1)$ point locations in the quadtree (querying a neighborhood of $25 \times 25$ around $C_j$). 
The fact that all cells in $\mathcal{P}(\S)$ lie in a range of $O(\log \psi)$ levels concludes the proof.
\end{proof}

\section{Maintaining and navigating quadtrees}
\label{sec:mat}

A compressed quadtree $T_X$ that stores a set $X$ of quadtree cells can be dynamically maintained in $O(\log |X|)$ time per insertion and deletion~\cite{har2010quadtrees}. 
Moreover, leaf location queries are supported in $O(\log |X|)$ time, which take as input some point $q \in \mathbb{R}^2$ and output the leaf of $T_X$ that contains $q$.
By Theorem~\ref{thm:quadtree_size}, $|X| = O(n \log \psi)$ in our setting. Since we assume that $\psi \in O(n^c)$, see Section~\ref{sec:results}, we can say that our insertion, deletion and point location operations in the compressed quadtree take $O(\log n)$ time. 
Compressed quadtrees additionally support level locations where for any query point $q \in \mathbb{R}^2$ and level $\ell$, the output is the quadtree cell at level $\ell$ that contains $q$; this can be used to dynamically maintain for all $\sigma \in \S$ the set $\mathcal{N}(\sigma)$ in our quadtree in $O(\log n)$ time per update in $\S$~\cite{buchin2011preprocessing}.

We maintain $L(\S)$ in $O(\log n)$ time per update to $\S$, while supporting point location queries.
We maintain in $O(\log n)$ time per update in $\S$ the values $d_{\max}$ and $d_{\min}$ that denote the maximal and minimal diameter in $\S$ respectively. We apply 3 more data structures:

\subparagraph{Marked Ancestor Trees (MAT).}
Alstrup, Husfeldt, and Rauhe~\cite{alstrup1998marked} introduce marked-ancestor trees.
Let $T$ be a dynamic tree.
Each node in $T$ is either marked or unmarked.
Given a node $v \in T$, the MAT supports changing the mark of $v$ or updating $T$ in $O(\log \log n)$ time.
Additionally, given a node $v \in T$, one can find the lowest/highest marked node on the path from $v$ to the root in $O(\log n / \log \log n)$ time.
 We augment our quadtree with a MAT.

\subparagraph{Orthogonal range trees.} Willard and Lueker~\cite{willard1985adding} show a data structure to store a set of $n$ squares using $O(n \log n)$
space. Given a query rectangle $\rho$, it can report an input square contained in $\rho$ in $O(\log^4 n)$ time, if such a square exists.
Given a query square $\rho$ it can report the number of squares that contain $\rho$ in $O(\log^4 n)$ time.
 We implement the range tree using general balanced
trees~\cite{andersson99gener_balan_trees}, so that we can support
updates in amortized $O(\log^4 n)$ time.

\subparagraph{Segment trees.} 
Segment trees~\cite{bkos2008} store a set of $n$ horizontal segments using $O(n\log n)$ space, so that for a vertical query segment $Q$ we obtain all $k$ input segments
that intersect $Q$ in $O(\log^2 n + k)$ time. The data structure can
again be made dynamic, supporting updates in $O(\log^2 n)$ amortized
time. For any $\sigma$, we store
the horizontal sides of $C_\sigma$ (scaled by a factor 5) in such an orthogonal intersection data
structure. Furthermore, we create a second such a data structure
storing all vertical sides.

\begin{theorem}
    \label{thm:square-maintenance}
        Let $\S$ be a set of axis-aligned squares. We augment $T(\S)$ with an $O(n \log n)$-size data structure
        using $O(n \log n)$ space, supporting inserting/deleting a square $\sigma$ in $O( |\mathcal{C}(\sigma)| \cdot \log^4 n + \log^6 n)$ time, and
    \begin{itemize}
        \item all cells in $\mathcal{C}(\S)$ are marked in our marked-ancestor tree;
                    \item for any query square $\gamma$, we can obtain
          $\mathcal{P}(\gamma)$ in $O(\log^2 n + |\mathcal{P}(\gamma)|)$ time.
          \item for any query cell $C$, we obtain the set $\mathcal{Z}(C) := \{ Z \mid C \in \mathcal{P}(\pi(Z)) \}$ in $O(\log^5 n)$ time. 
    \end{itemize}
\end{theorem}

\begin{proof}
By Theorem~\ref{thm:quadtree_size}, our quadtree requires $O(n \log \psi)$ space.
Using the standard operations on compressed quadtrees, we can maintain $\mathcal{N}(\S)$ in $O(\log n)$ time per update.
What remains for quadtree maintenance is to identify, decompress and mark all cells in $\mathcal{C}(\S)$. 

\textbf{Maintaining $\mathcal{C}(\S)$.}
Every cell $C \in \mathbb{F}$ has a counter that counts for how many $\sigma \in \S$, $C \in \mathcal{C}(\sigma)$. 
Whenever the counter is zero, we do not store it explicitly.
Otherwise, $C \in \mathcal{C}(\S)$ and we need to make sure that $C$ is decompressed and marked.
For each update of a square $\sigma$ with storing cell $C_\sigma$, there are two types of counter updates:
\begin{enumerate}
    \item updating counters of $C \in \mathcal{C}(\sigma)$, and
\item  updating the counters of $C_\sigma$ and its ancestors.
\end{enumerate}
We start with the first case. Instead of increasing counters, we do something slightly stronger as we can report all of $\mathcal{C}(\sigma)$. 
By definition, $C_\sigma \in \mathcal{C}(\sigma)$ and we add it to our output. 
We split $\sigma$ into eight rectangles that are bounded by $\sigma$ and the boundary of $C_\sigma$, see Figure~\ref{fig:containment-search}. 
We process each rectangle separately.
Consider such a rectangle $R$. 
We query our orthogonal range tree to report a storing cell $C_1$ in $R$ in  $O(\log^4 n)$ time. 
Given $C_1$, we walk in $O(\log \psi)$ time up the quadtree to find its largest ancestor $C_1'$ that is still contained in $\sigma$. 
By definition, $C_1' \in \mathcal{C}(\sigma)$ and we add it to our output.
Subsequently, we partition $R$ into nine rectangles that are bounded by the boundaries of $C_1'$ and recurse.  
For each cell in $\mathcal{C}(\sigma)$ we perform eight orthogonal range queries. 
For each range query, we either conclude that the range contains no cells in $\mathcal{C}(\sigma)$, or we identify at least one cell in $\mathcal{C}(\sigma)$.
As we recurse on rectangles that are bounded by cell boundaries and we explore all of $R$, we find all cells of $\mathcal{C}(\sigma)$ in $O( |\mathcal{C}(\sigma)| \cdot \log^2 n)$ time. As we find them, we may adjust their counters.

Now onto the second case, where we simply recompute all counters from scratch. 
There are at most $O(\log \psi)$ ancestors $C_\sigma, C_1, \ldots C_k$ of $C_\sigma$ that may be contained in a square in $\S$. 
For each of these, we recompute their counters from scratch. 
Fix an ancestor $C_i$ with parent $C_{i+1}$. 
By Lemma~\ref{lemma:uphill_containment}, there are at most $O(\log \psi)$ cells $Z$ such that $Z$ could store a square $\gamma$ with $C_i \in \mathcal{C}(\gamma)$. We obtain  all such $Z$ in $O(\log \psi \log n)$ time and iterate over each of them. 
For a fixed $Z$, we use the range tree to count how many $\gamma \in \pi(Z)$ contain $C_i$ in $O(\log^4 n)$ time. 
We then count how many $\gamma \in \pi(Z)$ contain $C_{i+1}$. 
The difference between these counts is the number of squares $\gamma^* \in \pi(Z)$ for which $C_i \in \mathcal{C}(\gamma^*)$. 
We compute and sum all these numbers to recompute the count of $C_i$. 
It follows from the fact that $O(\log \psi) \subset O(\log n)$ that we can maintain $\mathcal{C}(\S)$ in $O( |\mathcal{C}(\sigma)| \cdot \log^2 n + \log^6 n)$ time.

\textbf{Querying for $\mathcal{P}(\gamma)$.}
We show how to obtain for any query squares $\gamma$ the set
$\mathcal{P}(\gamma)$ in $O(\log^2 n + |\mathcal{P}(\gamma)|)$ time.
For each storing cell $C$, we consider the cell $C^*$ that is $C$ scaled by a factor 5 around its center. 
We store each of the boundary segments of $C^*$ in our Segment Tree. 
There is exactly one storing cell per square in $\S$, so maintaining the
Segment Tree intersection data structure takes $O(\log^2 n)$ time per
update and uses $O(n\log n)$ space. 
For any query square $\gamma$, we have $C \in \mathcal{P}(\gamma)$ if and only if one of the boundary segments of $\gamma$ intersects one of the boundary segments of $C^*$.
Thus, we can immediately use the intersection data structure to
compute $\mathcal{P}(\gamma)$ in $O(\log^2 n + |\mathcal{P}(\gamma)|)$
time, as each cell in $\mathcal{P}(\gamma)$ is reported at most a
constant number of times. 

\textbf{Querying for $\mathcal{Z}(C)$.}
Denote by $Z$ any cell, such that there could exist a $\gamma$ stored in $Z$ with $C \in \mathcal{P}(\gamma)$. By Lemma~\ref{lemma:uphill_permeter}, we can report all at $O(\log \psi)$ such cells in $O(\log \psi \log n)$ time.
For every such $Z$, we conceptually rotate the plane such that $Z$ lies above $C$.
Let $C^*$ be the cell $C$ increased by a factor 5 around its center. 
Any square $\rho$ in $\pi(Z)$ intersects $C^*$ in its boundary if and only if one of two conditions hold:
$\rho$ contains the top left endpoint of $C^*$ but not the bottom left endpoint, or
$\rho$ contains the top right endpoint of $C^*$ but not the bottom right endpoint. 
We select the top right endpoint of $C^*$ and we count how many squares in $\pi(Z)$ contain  the top right endpoint in $O(\log^4 n)$ time. 
We do the same for the bottom right endpoint. 
If the counts differ, there is at least one square in $\pi(Z)$ that intersects $C^*$ in its boundary and thus $C \in \mathcal{P}(\pi(Z))$.
Doing this for all $O(\log \psi)$ levels takes $O(\log^5 n)$ time.
\end{proof}

\begin{figure}[h]
  \centering
  \includegraphics[]{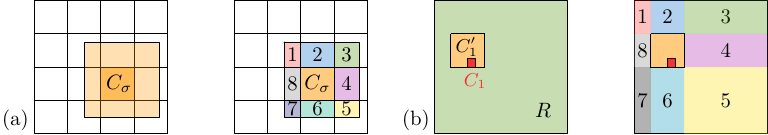}
  \caption{
      (a) Given a storing cell $C_\sigma$, we partition $\sigma$ into nine rectangles (one being $C_\sigma$). 
    (b) For each rectangle $R$, we do a range query to find a storing cell $C_1$ (if it exists). For the largest ancestor $C_1' \subset \sigma$ of $C_1$, we partition $R$ into nine rectangles once again and recurse.
  }
  \label{fig:containment-search}
\end{figure}

\section{Specific square intersection data structures}
\label{sec:Envelope_data_structures}

In this section we develop a solution for the following data structure problem: Let $\RR$ be a set of $m$ squares.  Let $R_1, \ldots, R_k$, be $k$ subsets of \RR that we refer to as \emph{conflict sets} and let $\ell \leq k$ be the maximum number of conflict sets that any square from \RR appears in. We want to store $\RR$ and all the conflict sets $R_1, \dots, R_k$ in a data structure that has size near linear in $m$ and $z = \sum_i |R_i|$, and support the following operations in the following time:

\begin{description}
\item[Insert$(\rho)$, ($O(l \cdot \log^3 m)$ time):] Insert a square $\rho$ into $\RR$.
\item[Delete$(\rho)$, ($O(l \cdot \log^3 m)$ time):] Delete a square $\rho$ from $\RR$ and every $R_i$ that it occurs in.
\item[Insert$(\rho, R_i)$, ($O( \log^3 m)$ time):] Insert a square $\rho$ in the conflict set $R_i$. If $R_i = \emptyset$, create a new conflict set.
\item[Delete$(\rho, R_i)$, ($O( \log^3 m)$ time):] Delete a square $\rho$ from the conflict set $R_i$. 
\item[Query$(\sigma, R_i, C)$, ($O( \log^3 m)$ time):] Given a query $\sigma$ whose center lies below all centers of all squares in \RR, and a horizontal line segment $C$ below $\sigma$, return (if it exists) a square $\rho \in \RR$ that
  intersects $\sigma$, but is not in the conflict set $R_i$, and that does not contain $C$.
\end{description}

In Section~\ref{app:squareintersection} we solve this problem as follows:
we map every square
$\rho = [\ell_\rho,r_\rho] \times [b_\rho,t_\rho] \in \RR$ to a point
$p_\rho = (b_\rho,\ell_\rho,r_\rho)$ in $\R^3$, and store these points
in a 3D-range tree $T$ augmented for range counting
queries~\cite{bkos2008}. Hence, every third-level subtree $T_\nu$
stores the number of points $m_\nu$ in $T_\nu$. Furthermore, for each
such subtree, and each conflict set $R_i$, consider the subset of
points stored in the leaves of $T_\nu$ for which the corresponding
square appears in $R_i$. If this set is non-empty then node $\nu$ also stores the
size
$m_{\nu,i} = |\{ p_\rho \mid p_\rho \in T_\nu \land \rho \in R_i \}|$
of this set. Furthermore, we maintain a bipartite graph between the squares in $\RR$ and the conflict sets $R_i$, so that given a
square $\rho \in \RR$ we can find the $\ell$ conflict sets it appears in
in $O(\ell)$ time. We implement all trees using general balanced
trees~\cite{andersson99gener_balan_trees} so that we can perform
updates efficiently. A 3D-range tree uses $O(m\log^2 m)$ space. Each
square in the multiset $\bigcup_i R_i$ contributes to $O(\log^3 m)$
nodes of $T$, and hence the entire structure uses at most
$O( (m + z) \log^3 m)$ space (Lemma~\ref{lem:lower_envelope_L1}). 
In Section~\ref{sub:maintaining_the_matching} we use this structure for connectivity queries.

\section{Square intersection data structure}
\label{app:squareintersection}

Let $\RR$ be a set of $m$ squares. Let $R_1, \ldots, R_k$, be $k$ subsets of \RR that we refer to as \emph{conflict sets} and let $\ell \leq k$ be the maximum number of conflict sets that any square from \RR appears in. We want to store $\RR$ and all the conflict sets $R_1, \dots, R_k$ in a data structure that has size near linear in $m$ and $z = \sum_i |R_i|$, and support the following operations:

\begin{description}
\item[Insert$(\rho)$:] Insert a square $\rho$ into $\RR$.
\item[Delete$(\rho)$:] Delete a square $\rho$ from $\RR$ and every $R_i$ that it occurs in.
\item[Insert$(\rho, R_i)$:] Insert a square $\rho$ in the conflict set $R_i$. If $R_i = \emptyset$, create a new conflict set.
\item[Delete$(\rho, R_i)$:] Delete a square $\rho$ from the conflict set $R_i$. If $R_i$ becomes empty, delete $R_i$.
\item[Query$(\sigma, R_i, C)$:] Given a query square $\sigma$ whose center lies below all centers of all squares in \RR, and a horizontal line segment $C$ below $\sigma$, return (if it exists) a square $\rho \in \RR$ that
  intersects $\sigma$, but is not in the conflict set $R_i$, and that
  also does not contain $C$.
\end{description}

We map every square
$\rho = [\ell_\rho,r_\rho] \times [b_\rho,t_\rho] \in \RR$ to a point
$p_\rho = (b_\rho,\ell_\rho,r_\rho)$ in $\R^3$, and store these points
in a 3D-range tree $T$ augmented for range counting
queries~\cite{bkos2008}. Hence, every third-level subtree $T_\nu$
stores the number of points $m_\nu$ in $T_\nu$. Furthermore, for each
such subtree, and each conflict set $R_i$, consider the subset of
points stored in the leaves of $T_\nu$ for which the corresponding
square appears in $R_i$. If this set is non-empty then node $\nu$ also stores the
size
$m_{\nu,i} = |\{ p_\rho \mid p_\rho \in T_\nu \land \rho \in R_i \}|$
of this set. Furthermore, we maintain a bipartite graph between the squares in $\RR$ and the conflict sets $R_i$, so that given a
square $\rho \in \RR$ we can find the $\ell$ conflict sets it appears in
in $O(\ell)$ time. We implement all trees using general balanced
trees~\cite{andersson99gener_balan_trees} so that we can perform
updates efficiently. A 3D-range tree uses $O(m\log^2 m)$ space. Each
square in the multiset $\bigcup_i R_i$ contributes to $O(\log^3 m)$
nodes of $T$, and hence the entire structure uses at most
$O( (m + z) \log^3 m)$ space. 
We show how to answer our queries and how to update the data structure:

\begin{restatable}{lemma}{lowerLone}
  \label{lem:lower_envelope_L1}
  Let \RR be a set of $m$ squares, let
  $z = \sum |R_i|$, and let there be at most $k \leq m$ conflict
  sets. Let each $\rho \in \RR$ appear in at most $l$ conflict
  sets. There is a data structure $\D^*(\RR)$ of size $O( (m + z) \log^3 m)$ that supports:
\begin{description}
\item[Insert$(\rho)$] in $O(l \cdot \log^3 m)$ amortized deterministic time,
\item[Delete$(\rho)$] in  $O(l \cdot \log^3 m)$ amortized deterministic  time,
\item[Insert$(\rho, R_i)$] in $O(\log^3 m)$ amortized deterministic  time, 
\item[Delete$(\rho, R_i)$] in $O(\log^3 m)$ amortized deterministic  time, and
\item[Query$(\sigma, R_i, C)$] in $O(\log^3 m)$ amortized deterministic  time.
\end{description}
\end{restatable}

\begin{proof}
  To insert a square $\rho \in \RR$ we use the standard insertion
  procedure for (dynamic) 3D-range trees; we insert the point $p_\rho$
  into $O(\log^3 m)$ subtrees. If, one of our subtrees becomes too
  unbalanced, we rebuild it from scratch. Rebuilding a $d$D-range tree
  on a set $P$ of $n$ points can be done in $O(n\log^{d-1} n)$
  time. However, we also still have to update the $m_{\nu,i}$ counts for
  each ternary subtree $T_\nu$ and each conflict list. We can do this
  in $O(ln\log^d n)$ time as follows. For each point $p_\rho \in P$ we
  obtain the at most $l$ conflict sets $R_i$ it appears in, and for
  each leaf corresponding to $p_\rho$ we simply walk upward updating
  the $m_{\nu,i}$ counts appropriately. It follows that the amortized
  insertion time is $O(l \log^3 m)$. Deletions are handled similarly
  in $O(l \log^3 m)$ amortized time.

  To insert or delete a square $\rho$ in one of the conflict sets
  $R_i$ we update the $m_{\nu,i}$ counts in the $O(\log^3 m)$ affected
  nodes (and we insert or delete the appropriate edge in the bipartite
  graph). If one of the $m_{\nu,i}$ counts reaches zero after a deletion, we stop storing it. Any counts that we do not store explicitly are considered to be zero.

  \begin{figure}[tb]
    \centering
    \includegraphics{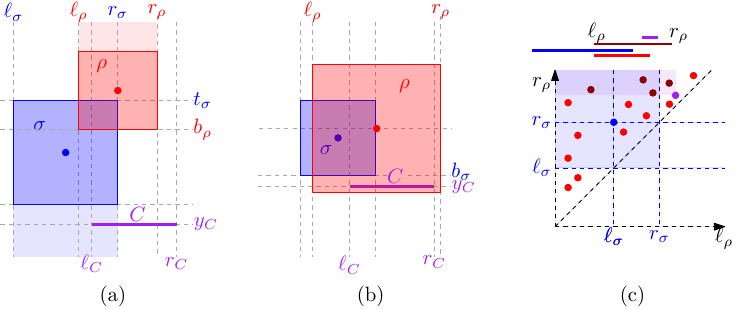}
    \caption{(a) Since the center of $\rho$ lies above the center of
      $\sigma$, we can essentially treat $\rho$ as a rectangle
      unbounded from the top, and $\sigma$ as a rectangle unbounded
      from the bottom. (b) A square $\rho$ may intersect $\sigma$ but
      is not allowed to contain $C$. (c) The $x$-extents of the
      objects map to points in $\R^2$. Squares (whose $x$-extent)
      intersects (the $x$-extent) of $\sigma$ lie in the blue region,
      and are not allowed to lie in the purple region.  }
    \label{fig:linfty_rangetree}
  \end{figure}

  Consider a query with square
  $\sigma = [\ell_\sigma,r_\sigma] \times [b_\sigma,t_\sigma]$,
  horizontal segment $C = [\ell_C,r_C] \times \{y_C\}$, and conflict
  set $R_i$ (see also Figure~\ref{fig:linfty_rangetree}). We will argue that there
  are $O(1)$ axis parallel boxes $Q_1,..,Q_{O(1)}$ such that the
  subset of squares from $\RR$ that intersect $\sigma$ but do not
  contain $C$ is the subset of points that lies in $\bigcup_j Q_j$. Our
  range tree allows us to obtain $O(\log^3 m)$ ternary subtrees
  $T_\nu$ that together represent the points in this region (in
  $O(\log^3 m)$ time). For each such subtree we then consider the
  counts $m_\nu$ and $m_{\nu,i}$: if they are equal all points
  (squares) in $T_\nu$ also appear in $R_i$, and hence there are no
  candidate points (squares) to be found in $T_\nu$. Otherwise, we
  have $m_\nu > m_{\nu,i}$, and hence $T_\nu$ does contain a
  point $p_\rho$ for which $\rho$ intersects $\sigma$, does not
  contain $C$, and for which $\rho \not\in R_i$. Moreover, one of the
  two children of $\nu$, say node $\mu$, must then also have
  $m_\mu > m_{\mu,i}$. This way we can find $p_\rho$ in time
  proportional to the height of $T_\nu$. It follows that the total query time
  is $O(\log^3 m)$. All that remains is to describe the regions
  $Q_1,..,Q_{O(1)}$.

  Since the center of $\sigma$ is guaranteed to lie below all centers
  of squares in $\RR$, and $y_C \leq b_\sigma$, we can essentially
  treat all squares as three-sided rectangles. In particular, a square
   $\rho \in \RR$ intersects $\sigma$ if and only if
  $p_\rho = (b_\rho,\ell_\rho,r_\rho)$ lies in the query range
  $Q = (-\infty,t_\sigma] \times (-\infty,r_\sigma] \times
  [\ell_\sigma, \infty)$ (see Figure~\ref{fig:linfty_rangetree}). Using that
  $y_c \leq b_\sigma \leq (b_\rho+t_\rho)/2$, we find that
  $C \subset \rho$ if and only if $p_\rho$ lies in the range
  $Q' = (\infty,y_C] \times (-\infty,\ell_C] \times [r_C,
  \infty)$. Hence, $\rho$ intersects $\sigma$, but does not contain
  $C$ if and only if $p_\rho \in Q \setminus Q'$. Since both $Q$ and
  $Q'$ are orthogonal boxes this region can be expressed as the union
  of $O(1)$ orthogonal ranges.
\end{proof}

\section{Maximal Bichromatic Matchings}
\label{sub:maintaining_the_matching}

The Maximal Bichromatic Matching data structure (MBM) in~\cite{kaplan2022dynamic} relies upon a square intersection data structure $\D$.
Consider a pair of disjoint quadtree cells $C_1, C_2$ and two sets $R \subseteq \pi(C_1)$ and $B \subseteq \pi(C_2)$. 
The MBM stores two square intersection data structures: $\D(R)$ and $\D(B)$, plus a maximal bichromatic matching $M_{RB}$ of the graph $G[R \cup B]$.

Given two such cells $C_1, C_2$, they dynamically maintain the matching as follows:
For all edges in $M_{RB}$, dynamically remove the endpoints from the square intersection structures storing $\D(R \backslash M_{RB})$ and $\D(B \backslash M_{RB})$. 
When a new square $\sigma$ gets inserted into $B$, query $\D(R \backslash M_{RB})$ to find a square in $R \backslash M_{RB}$ that intersects $\sigma$. If such a square $\rho$ exists, add the edge $(\rho, \sigma)$ to the matching $M_{RB}$. Subsequently delete $\rho$ from $\D(R \backslash M_{RB})$.
With a similar procedure for deletions, one can dynamically maintain $M_{RB}$ in time proportional to the update and query time of the intersection data structure $\D$. 

\begin{figure}[b]
  \centering
  \includegraphics[width=\linewidth]{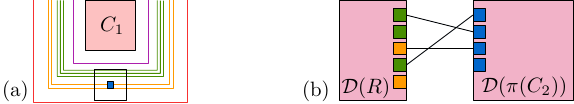}
  \caption{
(a) $C_1$ with a blue $C_2 \in \mathbb{P}^*(\pi(C_1))$. 
    The elements of the set $R = \{ \sigma \in \pi(C_1) \mid C_2 \in \mathbb{P}^*(\sigma) \}$
    are the green and yellow squares. 
(b) 
If we split $C_2$ to create a cell $C'$, then the corresponding $R'$
would consist only of the orange squares. 
Since there exists no efficient way to split intersection data structures, constructing the new data structure on $R'$ takes linear time. 
\vspace{-0.3cm}
  }
  \label{fig:mbm}
\end{figure}

\subparagraph{Defining the sets $R$ and $B$.}
The sets $R$ and $B$ must be carefully chosen if we want to avoid spending quadratic space in $n$ or $\psi$. 
Recall the pipeline of Kaplan \etal~\cite{kaplan2022dynamic}. 
For each storing cell $C_1$, for each $C_2 \in \mathbb{P}^*( \pi(C_1))$, they store an MBM between the pair $(C_1, C_2)$.
We know that there may be $\Theta(\psi^2)$ cells in the set  $\mathbb{P}^*( \pi(C_1))$. 
Suppose that for each pair $C_1, C_2$ they set $R \gets \pi(C_1)$ and $B \gets \pi(C_2)$. 
Every square in $\pi(C_1)$ may get stored $O(\psi^2)$ times and the total space usage is $O(\psi^2 n)$. 
To improve space and time usage, the authors of~\cite{kaplan2022dynamic} instead set $R \gets \{ \sigma \in \pi(C_1) \mid C_2 \in \mathbb{P}^*(\sigma) \}$ and $B \gets  \pi(C_2)$, see Figure~\ref{fig:mbm}(a).
Since each square $\sigma \in \pi(C_1)$ has $O(\psi)$ cells in its perimeter $\mathbb{P}^*(\sigma)$, each square $\sigma$ is stored $O(\psi)$ times and the total space is $O(\psi \cdot |\pi(C_1)|)$.
A charging argument then shows that the total space required is $O(\psi n)$. 
The data structures can be updated in $O(\psi)$ times: the update time of the intersection data structures $\D(R \backslash M_{RB})$ and $\D(B \backslash M_{RB})$. 

\subparagraph{Defining an MBM for adaptive $\psi$.}
When the aspect ratio is adaptive (or, even when it is bounded), the approach by Kaplan~\etal~\cite{kaplan2022dynamic} requires an update time linear in $\psi n$, since, after increasing $\psi$, the perimeter $\mathbb{P}^*(\pi(\sigma))$ increases in size by $O(\psi)$ for every storing cell $C_1$.
We could try to avoid this issue by replacing their definition of perimeter with ours. 
That is, for every cell $C_1$, we would consider the $O(\psi^2)$ cells $C_2$ in $\mathcal{P}(\pi(C_1))$.
For the pair $C_1, C_2$, we want to maintain an MBM between sets $R \gets \{ \sigma \in \pi(C_1) \mid C_2 \in \mathcal{P}(\sigma) \}$ and $B \gets  \pi(C_2)$; by storing $R$ and $B$ each in their separate data structure for square intersection queries. 
However, such a structure can also not be efficiently dynamically maintained, see Figure~\ref{fig:mbm}(b).
Thus, we must avoid storing the sets $R\backslash M_{RB}$ and $B \backslash M_{RB}$ explicitly in a data structure.

\subparagraph{Applying our data structure problem}
We now apply our previous data structure problem.
Let $C_1$ be a storing cell in our quadtree $T(\S)$.
We maintain the data structure $\D^*(\pi(C_1))$ (i.e. we set $\RR \gets \pi(C_1)$).
Let there be $k$ cells in the perimeter $\mathcal{P}(\pi(C_1))$. Then
by Lemma~\ref{lemma:perimeter_size}, we have $k \in O(\psi^2)$.
Let $C_i$ be some storing cell in $\mathcal{P}(\pi(C_1))$, we denote by $M_{i}$ some maximal matching in the graph $G[ R \cup B]$ for $R \gets \{ \gamma \in \pi(C_1) \mid C_i \in \mathcal{P}(\gamma) \} $ and $B \gets \pi(C_i)$.
Denote by $R_i$ the squares in $\pi(C_1)$ that are part of $M_i$; we say that $R_i$ is a conflict set.
The result of this transformation are $k$ conflict sets of $\pi(C_1)$.
Moreover, each square $\rho$ in $\pi(C_1)$ may appear in at most $l \in O(|\mathcal{P}(\rho)|) \subset O(\psi)$ conflict sets.
We now apply Lemma~\ref{lem:lower_envelope_L1} four times (one for each direction).
We maintain for all $C_1$ this data structure $\D^*( \pi(C_1) )$ and show:

\begin{theorem}
\label{thm:mbm_L1}
     Let $\S$ be a set of $n$ squares with adaptive aspect ratio $\psi$.
   We can maintain for pair of storing cells $(C_1, C_2)$ a maximal bichromatic matching in $G[ R \cup \pi(C_2)]$ with $R \gets  \{ \gamma \in \pi(C_1) \mid C_2 \in \mathcal{P}(\gamma) \}$.
Our solution uses $O(n \log \psi \log^3 n)$ space. Inserting/deleting a square $\sigma$ requires $O( |\mathcal{P}(\sigma)| \cdot \log^3 n + \log^5 n)$ amortized  time.
\end{theorem}

\begin{proof}
  We first analyse our space usage and then show how to maintain each matching.

\textbf{Upper bounding size.}
For $C_1$, denote by $z_1$ the sum of all $C_i$, over all edges in the maximal matching of $G[R \cup \pi(C_i)]$ with $R \gets \{ \gamma \in \pi(C_1) \mid C_2 \in \mathbb{P}^*(\gamma) \}$.
Lemma~\ref{lem:lower_envelope_L1} presents a data structure with size $O((|\pi(C_1)| + z_1) \log^3 |\pi(C_1)|)$.
Denote by $\mathcal{M}(\S)$ the set of all edges, across all maximal bichromatic matchings, for all pairs of storing cells $(C_1, C_2)$.
Let $\mathcal{M}(\S)$ contain $z^*$ elements.
It follows that all these data structures use at most $O( (n + z^*) \cdot \log^3 n)$ total space. 
We upper bound the number of edges in $z^*$ by charging each edge to one of their endpoints.  
Intuitively, we charge each matched edge to the squares of the smallest quadtree cell.  Every square $\sigma \in \S$ receives at most $O(\log \psi)$ charges and $z^*$ is upper bound by $n \log \psi$. 

More formally, we over-estimate the edges in $\mathcal{M}(\S)$. 
Fix for \emph{every} pair $(C_1, C_2)$ with $C_2 \in \mathcal{P}(\pi(C_1)))$ an arbitrary maximal bichromatic matching in the graph $G[\pi(C_1) \cup \pi(C_2)]$ (i.e., we ignore the fact that we match between sets $R \subseteq \pi(C_1)$  and $B \subseteq \pi(C_2)$, and fix some potentially larger matching in the bigger graph $G[\pi(C_1) \cup \pi(C_2)]$). 
Denote for $C_1$ by $\mathcal{M}_\prec(C_1)$ the set of all matchings between $C_1$ and $C_2$ where $|C_1| \prec |C_2|$ for $\prec~\in \{ <, =, > \}$. 
Any edge $e \in \mathcal{M}(\S)$ is in $\mathcal{M}_=(C_1) \cup \mathcal{M}_<(C_1)$ for some storing cell $C_1$. 
First, we upper bound $|\mathcal{M}_=(C_1)|$. There are $O(1)$ cells $C_2$ with $|C_1|=|C_2|$ and $C_1 \in \mathcal{P}(\pi(C_2))$ or vice versa. 
For every such $C_2$, there can be at most $O(|\pi(C_1)|)$ edges in a MBM in $G[ \pi(C_1) \cup \pi(C_2)]$. Thus, there are at most $O( |\pi(C_1)|)$ edges in $\mathcal{M}_=(C_1)$. 
Second, by Lemma~\ref{lemma:uphill_permeter}, there are at most $O(\log \psi)$ cells $C_2$ with $C_1 \in \mathcal{P}(\pi(C_2))$ and $|C_1| < |C_2|$. 
Again, every matching in   $G[ \pi(C_1) \cup \pi(C_2)]$ has at most $O(|\pi(C_1)|)$ edges, thus $\mathcal{M}_<(C_1)$ contains at most $O(|\pi(C_1)| \cdot \log \psi)$ edges. Now:
\[
z^* = | \mathcal{M}(\S)| \leq \sum_{\textnormal{storing cell } C_1} |\mathbb{M}_=(C_1)|  + |\mathbb{M}_<(C_1)| \leq \sum_{\textnormal{storing cell } C_1} |\pi(C_1)| \cdot \log \psi \leq n \log \psi
\]
It follows we use at most $O( (n + z^*) \cdot \log^3 n)  \subset O( n \log \psi \log^3 n))$ space.

    \textbf{Maintaining the MBM.}
  Suppose that we delete a square $\sigma$ from $\S$ (this is the more difficult case). 
    We can find its storing cell $C_\sigma$ in $O(\log n)$ time using standard quadtree navigation.
We obtain $\D^*( \pi(C_\sigma))$ with its $k \in O(\psi^2)$ conflict sets.
Recall that $\sigma$ appears in at most $l \in O(|\mathcal{P}(\sigma)|)$ conflict sets. 
By Lemma~\ref{lem:lower_envelope_L1}, we may remove $\sigma$ from the data structure  $\D^*(\pi(C_\sigma))$  in $O(l \log^2 n) \subseteq O(|\mathcal{P}(\sigma)| \cdot \log^2 n)$ amortized time. 
What remains is to update all the matchings. 
We recall that we maintain a matching between $(C_\sigma, C_2)$ in two cases:
either the cell $C_2 \in \mathcal{P}(\pi(C_\sigma))$ or $C_\sigma \in \mathcal{P}(\pi(C_2))$.
There are at most $O(|\mathcal{P}(\pi(C_\sigma))|)$ cells of the first case, and $O(\log \psi)$ cells of the second case. 
By Theorem~\ref{thm:square-maintenance}, we may obtain all such cells in $O(\log^5 n + |\mathcal{P}(\sigma)|)$ time.

\textbf{Processing a cell $C_2$.}
Fix a cell $C_2$ with a corresponding conflict set $R_2$ in $\D^*(\pi(C_\sigma))$.
We test if $\sigma$ was an endpoint of the matching $(\sigma, \rho)$ by searching over the conflict set $R_2$. 
If so, then we delete $\sigma$ from the conflict set $R_2$.
What remains is to try and rematch $\rho$.

Thus, we want to find a square in $R = \{ \gamma \in \pi(\sigma) \mid C_2 \in \mathcal{P}(\gamma) \}$, that is not already in the conflict set $R_2$ (i.e. not already in the matching between $G[R \cup B]$). 
Denote by $K$ the cell $C_\sigma$ scaled by a factor 5 around its center and by $\underline{K}$ the bottom facet of $K$. 
We claim that $\rho$ can be matched to a square in $R$ if and only if Query$(\rho, R_2, \underline{K})$ from Lemma~\ref{lem:lower_envelope_L1} is not empty. 

Indeed, for any $\gamma \in \pi(C_\sigma)$ that intersects $\rho$ and contains $\underline{K}$, must contain $K$. By definition, $C_2 \not \in \mathcal{P}(\gamma)$ and thus $\gamma \not \in R$. 
For any $\gamma \in \pi(C_\sigma)$ that intersects $\rho$ where $\gamma \in R_2$, by definition $\gamma \not \in R \backslash M_{RB}$. 
For any $\gamma \in \pi(C_\sigma)$ that does not intersect $\rho$, there is no edge between $\gamma$ and $\rho$ in $G[R \cup B]$.
It follows that with one query we may rematch $\rho$ in $O(\log^3 n)$ amortized time.

Since there are at most $O( |\mathcal{P}(\sigma)| + \log \psi)$ cells $C_2$ to consider, we can maintain every MBM in $O( |\mathcal{P}(\sigma)| \cdot \log^3 n + \log^5 n)$ time.
\end{proof}

\section{Dynamic connectivity in square intersection graphs}

Having formally introduced and analysed every component, we can now fully state what our data structure maintains. For an illustration, we refer back to Figure~\ref{fig:newpipeline}.
We store a data structure that uses at most $O(n \log^3 n \log \psi)$ space:
\begin{enumerate}[(1)]
    \item We store  $\S$ in a quadtree $T(\S)$.
\begin{itemize}
\item This quadtree contains  for each cell $\sigma \in \S$ the neighborhood $\mathcal{N}(\sigma)$. Additionally, we
\item maintain all $C \in \mathcal{C}(\S)$ with $O(| \mathcal{C}(\sigma)| \cdot \log^4 n + \log^6 n)$ amortized time (Thm~\ref{thm:square-maintenance}).
\item This quadtree requires $O(n \log \psi)$ space (Thm~\ref{thm:square-maintenance}).
\end{itemize}
\item  We augment our quadtree with a Marked-Ancestor Tree (MAT).
\begin{itemize}
    \item We mark each cell $C \in \mathcal{C}(\S)$ in the MAT (Thm~\ref{thm:square-maintenance}).
\end{itemize}
\item For any storing cell $C$, we define a conflict set $R_i$ for all cells $C_i \in \mathcal{P}(\pi(C))$.
We store $\pi(C)$ with the conflict sets in our square intersection data structure  $\D^*(\pi(C))$.
\begin{itemize}
    \item Let $z^* = \sum_C \sum_{i} |R_i|$, the total space required is $O( (n + z^*) \log^3 n)$ (Lem~\ref{lem:lower_envelope_L1}).
    \item By the proof of Theorem~\ref{thm:mbm_L1}, $z^* \in O(n\log \psi)$ so we use $O(n \log \psi \log^3 n)$ total space. 
    \end{itemize}
\item  For each storing cell $C_1$ and each $C_2 \in \mathcal{P}(\pi(C_1)) $, we store a Maximal Bichromatic Matching (MBM) in $G[R \cup B]$.
\begin{itemize}
    \item We set $R$ as the set of squares in $C_1$ that have $C_2$ in their perimeter ($R \gets  \{ \gamma \in \pi(C_1) \mid C_2 \in \mathcal{P}(\gamma) \}$ and $B \gets \pi(C_2)$. 
\item Updates in $\S$ require $O( |\mathcal{P}(\sigma)| \cdot \log^3 n + \log^5 n)$ amortized time (Thm~\ref{thm:mbm_L1}). 
\end{itemize}
\item Finally, for any pair $(C_1, C_2)$, if their MBM is not empty, we store an edge between them. 
\begin{itemize}
\item We maintain the resulting `proxy graph' in the HLT data structure \cite{holm2001poly}. 
\item Inserting or deleting a square $\sigma$ introduces at most $O(|\mathcal{P}(\sigma)| + \log \psi)$ new edges.
\end{itemize}
\end{enumerate}

\noindent
We finally show that this data structure implies the following:

\begin{restatable}{theorem}{Lone}  
\label{thm:L1}
Let $\S$ be a set of squares with adaptive aspect ratio $\psi$. 
We can store $\S$ in a dynamic data structure of size $O(n \log^3 n \log \psi)$ with $O( (|\mathcal{C}(\sigma)| + |\mathcal{P}(\sigma)|)\log^4 n + \log^6 n)$ amortized deterministic update time such that for any pair of squares $(\sigma, \rho)$ we can query for the connectivity between $\sigma$ and $\rho$ in $O(\log n / \log \log n)$ time. 
\end{restatable}

\begin{proof}
    Our pipeline functions identical to the pipeline of~\cite{kaplan2022dynamic}. 
    Given $\sigma$, we obtain a pointer to its storing cell $C_\sigma$ in $O(1)$ time. 
    We then query the marked-ancestor tree in $O(\log n / \log \log n)$ time to find the largest ancestor $C_\alpha$ of $C_\sigma$ that is marked. The cell $C_\alpha$ is marked by at least one squares $\gamma$ that contains $C_\alpha$ in its interior. We obtain a pointer to $\gamma$ and its storing cell $C^*$ in $O(1)$ time. We note that if there is also some squares $\gamma'$ that marked $C_\alpha$, we may arbitrarily get a pointer to either $\gamma$ or $\gamma'$. 
    Doing the same procedure for $\rho$ gives a cell $R^*$. 
    We test whether $C^*$ and $R^*$ are connected in the proxy graph $O(\log n / \log \log n)$ time. 
    We now claim that these two cells are connected in the proxy graph if and only if $(\rho, \sigma)$ are connected. 
    The key observation to prove this claim is that, if we were to rescale the plane, our graph contains the proxy graph maintained by  Kaplan~\etal~\cite{kaplan2022dynamic} as a subgraph.
    Indeed at the time of a query, $\psi$ is fixed.
    Thus, we may rescale the plane such that every square has a diameter in $[\frac{1}{\psi}, 1]$.
    Let $H(\S)$ be the quadtree of~\cite{kaplan2022dynamic}, then $T(\S) \subset H(\S)$. 
    Kaplan~\etal maintain for every pair $(C_1, C_2)$ with $C_2 \in \mathbb{P}^*(\pi(C_1))$ an Maximal Bichromatic Matching in the graph $G[R' \cup B']$ for $R' \gets  \{ \gamma \in \pi(C_1) \mid C_2 \in \mathbb{P}^*(\gamma) \}$ and $B' \gets \pi(C_2)$.
    For each nonempty MBM between a pair $(C_1, C_2)$, they store an edge in the proxy graph. 
    
    Note that if the MBM is nonempty, then both $C_1$ and $C_2$ are storing cells.
    It follows that $C_2 \in \mathcal{P}(\pi(C_1))$; and that $R' = R$.
    Thus, we store for each non-empty MBM a maximal bichromatic
    matching in the graph $G[R \cup B] = G[R' \cup B']$ as in~\cite{kaplan2022dynamic}.
    This implies that after rescaling, whenever there exists an edge in the proxy graph of~\cite{kaplan2022dynamic}, there exists an edge in our data structure. 
    Thus, we may immediately apply the proof of Theorem 4.3 in~\cite{kaplan2022dynamic} to conclude that $(\sigma, \rho)$ are connected if and only if $(C^*, R^*)$ are. 
\end{proof}

\bibliographystyle{plain}
\bibliography{refs}

\end{document}